%% file: iclr2026_conference.tex
\documentclass{article} 
\usepackage{iclr2026_conference,times}

\input{math_commands.tex}

\usepackage[colorlinks=true, linkcolor=blue, citecolor=blue, urlcolor=magenta]{hyperref}
\usepackage{url}
\usepackage{graphicx}
\usepackage{csquotes} 
\usepackage{amsthm} 
\usepackage{amssymb}
\usepackage{subcaption}
\usepackage{color}
\usepackage{enumitem}
\usepackage{cleveref}
\usepackage{booktabs}
\usepackage{multirow}
\usepackage{authblk}

\usepackage{algorithm}
\usepackage{algorithmic}
\usepackage{amsmath}
\usepackage{tabularx}

\usepackage{amsthm}
\newtheorem{theorem}{Theorem}[section]

\newtheorem{proposition}[theorem]{Proposition}

\theoremstyle{definition}

\theoremstyle{remark}

\title{Redistributing Rewards Across Time and Agents for Multi-Agent Reinforcement Learning}

\author[1]{Aditya Kapoor\thanks{\texttt{Email: aditya.kapoor@postgrad.manchester.ac.uk}}}
\author[2]{Kale-ab Tessera}
\author[3]{Harshad Khadilkar}
\author[3]{Mayank Baranwal}
\author[4]{Jan Peters}
\author[2]{Stefano Albrecht}
\author[1]{Mingfei Sun}

\affil[1]{University of Manchester}
\affil[2]{University of Edinburgh}
\affil[3]{IIT, Bombay}
\affil[4]{TU Darmstadt}
\iclrfinalcopy 
\begin{document}

\maketitle

\begin{abstract}
Credit assignment—disentangling each agent's contribution to a shared reward—is a critical challenge in cooperative multi-agent reinforcement learning (MARL). To be effective, credit assignment methods must preserve the environment's optimal policy. Some recent approaches attempt this by enforcing return equivalence, where the sum of distributed rewards must equal the team reward. However, their guarantees are conditional on a learned model's regression accuracy, making them unreliable in practice. We introduce Temporal-Agent Reward Redistribution (TAR²), an approach that decouples credit modeling from this constraint. A neural network learns unnormalized contribution scores, while a separate, deterministic normalization step enforces return equivalence by construction. We demonstrate that this method is equivalent to a valid Potential-Based Reward Shaping (PBRS), which guarantees the optimal policy is preserved regardless of model accuracy. Empirically, on challenging SMACLite and Google Research Football (GRF) benchmarks, TAR² accelerates learning and achieves higher final performance than strong baselines. These results establish our method as an effective solution for the agent-temporal credit assignment problem.
\href{https://github.com/AdityaKapoor74/MARL_Agent_Temporal_Credit_Assignment}{Github Code}
\end{abstract}

\section{Introduction}
\label{sec:introduction}
MARL \citep{marl-book} is a powerful paradigm for solving complex cooperative tasks, with landmark successes in domains ranging from logistics and robotics to challenging games \citep{krnjaic2024scalable, sartoretti2019primal, Vinyals2019GrandmasterLI, kurach2020google}. In these settings, teams of agents must learn to synchronize their actions to achieve a shared goal.

Despite this progress, a key bottleneck is the multi-agent credit assignment problem: allocating a shared team reward to guide agent learning. The challenge is amplified in episodic MARL where a single feedback signal arrives only at a trajectory's end. Here, agents struggle to resolve two coupled problems: when their critical contributions occurred (temporal credit assignment) and which agents were responsible (agent credit assignment).

For credit assignment methods to be effective, they should not alter the environment's optimal policy. Some recent approaches, like STAS and AREL, attempt this by enforcing a return equivalence constraint, where distributed rewards must sum to the team reward. However, this approach is theoretically fragile They train a single model to both predict credit and satisfy this constraint simultaneously. Consequently, their policy-invariance guarantee is conditional on the network being a perfect regressor—a condition practically unattainable. Any prediction error breaks the guarantee and risks leading to suboptimal policies.

To solve this fragility, we introduce Temporal-Agent Reward Redistribution (TAR²), a method designed for structural robustness, which we illustrate in Figure~\ref{fig:tar2_architecture}. The core idea is to \textbf{decouple credit modeling from constraint satisfaction}. A neural network performs a single, focused task: learning unnormalized scores representing the \textit{absolute} importance of each agent's actions. A separate, deterministic normalization step then makes this importance \textit{relative}, converting the scores into a valid probability distribution used to construct the final rewards. This guarantees return equivalence by construction, irrespective of the network's accuracy. We demonstrate that this two-step architecture is equivalent to a valid Potential-Based Reward Shaping (PBRS) \citep{ng1999policy}, which provides the formal guarantee that the optimal policy is preserved.

\section{Background}
\label{sec:background}

We formalize the cooperative MARL problem as a Decentralized Partially Observable Markov Decision Process (Dec-POMDP)~\citep{Oliehoek2016ACI} and review Potential-Based Reward Shaping (PBRS)~\citep{ng1999policy}, the theoretical foundation for our approach.

\subsection{Problem Formulation}
\label{subsec:dec_pomdp}
A fully cooperative multi-agent task is a Dec-POMDP~\citep{Oliehoek2016ACI, amato2024partial}, defined by the tuple $\mathcal{M} = \langle \mathcal{S}, \{\mathcal{A}_i\}_{i=1}^{N}, \mathcal{P}, \{\Omega_i\}_{i=1}^{N}, \mathcal{O}, \mathcal{R}, \rho_0, \gamma, N \rangle$. The environment consists of $N$ agents interacting in a global state space $\mathcal{S}$ with an initial state distribution $\rho_0$. At each timestep $t$, the team takes a joint action $\mathbf{a}_t = \{a_{1,t}, \dots, a_{N,t}\}$ from the joint action space $\mathbf{\mathcal{A}} \! = \! \! \times_{i=1}^N \mathcal{A}_i$, which causes a transition to a new state $s_{t+1}$ according to $\mathcal{P}(s_{t+1} | s_t, \mathbf{a}_t)$. Due to partial observability, agent $i$ receives only a local observation $o_{i,t} \in \Omega_i$ from the observation function $\mathcal{O}$.

Because a single observation is often insufficient to disambiguate the true state, each agent must condition its policy on its local action-observation history, $\tau_{i,t} = (o_{i,0}, a_{i,0}, \dots, o_{i,t-1}, a_{i,t-1}, o_{i,t})$ \citep{amato2024partial}. For notational convenience, we denote the histories of all other agents as $\tau_{-i,t} = \{\tau_{j,t}\}_{j \neq i}$, and the joint history of all agents as $\boldsymbol{\tau}_t = (\tau_{1,t}, \dots, \tau_{N,t})$. In a decentralized execution setting, each agent learns a policy $\pi_i(a_{i,t} | \tau_{i,t})$ that depends only on its own history. 
The agents' policies combine to form a joint policy $\boldsymbol{\pi} = \prod_{i=1}^{N} \pi_i$. The team receives a shared reward $r_t = \mathcal{R}(s_t, \mathbf{a}_t)$ and, with a discount factor $\gamma$, aims to learn a joint policy $\boldsymbol{\pi}$ that maximizes the expected discounted return: $J(\boldsymbol{\pi}) = \mathbb{E}_{\tau \sim \boldsymbol{\pi}, \rho_0}\left[\sum_{t=0}^{T} \gamma^t r_t\right]$, where $\boldsymbol{\tau}$ denotes a full episode trajectory. To learn effective decentralized policies in large state-action spaces, we adopt the Centralized Training with Decentralized Execution (CTDE) paradigm \citep{foerster2018counterfactual, lowe2017multi}. During centralized training, agents leverage global information to guide learning, while at test-time they execute policies using only local information.

We focus on the challenging episodic setting, where reward is only dispensed at the end of an episode $r_t=0$ for $t<T$, and $r_t = R(s_T)$ is a terminal reward. This sparsity makes credit assignment exceptionally difficult, as agents must deduce which actions in a long trajectory led to the final outcome from an episodic signal. While policies are executed decentrally, we assume a centralized training paradigm where a shaping function can access the joint history $\boldsymbol{\tau}_t$ to guide learning.

\subsection{Potential-Based Reward Shaping}
\label{subsec:pbrs}
To create dense rewards without distorting the underlying problem, we leverage PBRS~\citep{ng1999policy}. In single-agent RL, PBRS augments the environment's reward $\mathcal{R}(s, a)$ by an additional shaping reward, $F(s, a, s') = \gamma\Phi(s') - \Phi(s)$, 
using a potential function $\Phi: \mathcal{S} \to \mathbb{R}$ to improve the speed of convergence. The key property of PBRS is that it guarantees \textit{policy invariance} -- any optimal policy in the shaped-reward MDP remains optimal in the original MDP. This guarantee extends to the fully cooperative multi-agent setting, preserving the set of optimal joint policies~\citep{devlin2014potential} and forming a theoretically sound basis for credit assignment.

\begin{figure*}[!t]
    \centering
    \includegraphics[width=1.\textwidth]{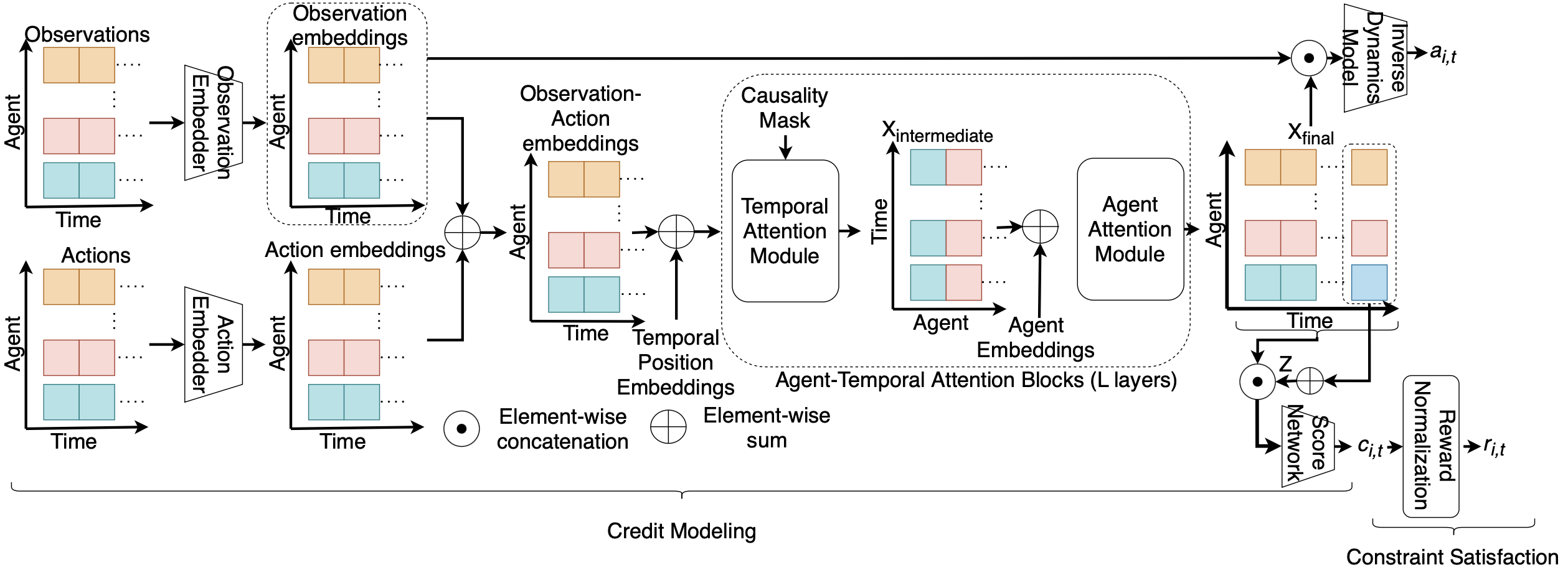}
    \caption{The TAR² architecture processes trajectory data through four main stages. (1) Input sequences are converted into embeddings with positional encoding. (2) A multi-layer transformer block with sequential Temporal and Agent Attention builds context-aware representations, regularized by an auxiliary Inverse Dynamics task to ensure causality. (3) The Score Network computes unnormalized scores by conditioning each timestep's representation on a learned Final Outcome Embedding (Z). (4) A final Probabilistic Normalization step converts these scores and the global reward $R(s_T)$ into dense, per-agent rewards $\{r_{i,t}\}$ that satisfy strict return equivalence.}
    \label{fig:tar2_architecture}
\end{figure*}

\section{Temporal-Agent Reward Redistribution}
\label{sec:approach}

As established, PBRS provides a sound theoretical basis for policy-preserving credit assignment. However, its practical application is not straightforward. The guarantee of policy invariance holds for \textit{any} potential function, but an unconstrained or poorly learned one can introduce significant noise and high variance into the policy gradients, destabilizing and slowing down convergence.

Our method, Temporal-Agent Reward Redistribution (TAR²), is designed to harness the guarantees of PBRS while mitigating these practical instabilities. The core idea is to decouple credit modeling from constraint satisfaction. As illustrated in Figure~\ref{fig:tar2_architecture}, a neural network performs a single, focused task: learning unnormalized contribution scores. A separate, deterministic normalization step then constructs the final rewards, guaranteeing return equivalence by construction. In this section, we formalize this design, prove its theoretical guarantees, analyze its learning dynamics, and detail the architecture that operationalizes these principles.

\subsection{Reward Redistribution Formulation}
\label{subsec:rew_red_form}
We formalize our two-step process as follows. The reward model, parameterized by $\theta$, learns unnormalized contribution scores, $c_{i,t}$. We then use a deterministic shift-and-normalize scheme to convert these scores into weights. The weights convert the scores to final shaped rewards, $r_{i,t}$, enforcing strict return equivalence ($\textstyle \sum_{t,i} r_{i,t} = R(s_T)$) by construction. This structural guarantee is critical for reducing the variance in the learning signal (Sec 3.4.1), distinguishing TAR² from prior work where policy invariance is a fragile and conditional property. First, we compute temporal weights,
\begin{equation}
    w_{t}^{\text{temp}} = \frac{c_{t}^{\text{agg}} - \min_{t'} c_{t'}^{\text{agg}}}{\sum_{t''}(c_{t''}^{\text{agg}} - \min_{t'} c_{t'}^{\text{agg}}) + \epsilon},
    \label{eq:temporal_norm}
\end{equation}
by normalizing aggregated scores ($c_{t}^{\text{agg}} = \textstyle \sum_i c_{i,t}$) across the trajectory.
We then compute agent-specific weights, 
\begin{equation}
    w_{i,t}^{\text{agent}} = \frac{c_{i,t} - \min_{j} c_{j,t}}{\sum_{k}(c_{k,t} - \min_{j} c_{j,t}) + \epsilon},
    \label{eq:agent_norm}
\end{equation}
by normalizing the individual scores within each timestep.
The indices iterate over active agents and timesteps, and $\epsilon$ (e.g. 1e-8) is a small constant for numerical stability. The final redistributed reward is then constructed as
\begin{equation}
s_{i,t} = w_{t}^{\text{temp}} w_{i,t}^{\text{agent}} R(s_T).
\label{eq:reward_redist}
\end{equation}

\subsection{Optimal Policy Preservation}
\label{subsec:pbrs_guarantee}

We prove TAR\textsuperscript{2} preserves the optimal policy by framing it within multi-agent PBRS. Let $r_{i,t}^{\text{orig}}$ be agent $i$'s ground-truth contribution to the team reward $r_t^{\text{orig}}$. Our model produces a dense, shaped reward $s_{i,t} = w_t^{\text{temp}} w_{i,t}^{\text{agent}} R(s_T)$. While the complete potential-based reward is $r'_{i,t} = s_{i,t} + r_{i,t}^{\text{orig}}$, TAR\textsuperscript{2} uses only the shaped reward $s_{i,t}$ for learning. This is a critical design choice: because $r_{i,t}^{\text{orig}}$ is zero for all $t<T$ in the episodic setting, dropping it eliminates a sparse, high-variance signal. As we mathematically justify in Sec 3.4.1, this significantly reduces variance and stabilizes learning.

\begin{proposition}[Optimal Policy Preservation]
Let $\mathcal{M}_{\text{env}}$ be a Dec-POMDP where agent $i$ receives reward $r_{i,t}^{\text{orig}}$. Let $\mathcal{M}_{\text{TAR}^2}$ be an identical environment where agent $i$ receives the augmented reward $r'_{i,t} = r_{i,t}^{\text{orig}} + s_{i,t}$. Any joint policy $\pi^*$ optimal in $\mathcal{M}_{\text{TAR}^2}$ is also optimal in $\mathcal{M}_{\text{env}}$.
\end{proposition}

\begin{proof}
The proof demonstrates that $s_{i,t}$ is a valid per-agent potential-based shaping reward. As established by \citet{Devlin2011TheoreticalCO}, using such individual potential functions preserves the set of optimal joint policies. For each agent $i$, we define a history-based potential function $\Phi_i(\tau_t) = \sum_{k=0}^{t-1} s_{i,k}$ (assuming $\gamma=1$ for the episodic setting). The corresponding shaping function is:
\begin{equation}
    F_{i,t} = \gamma \Phi_i(\tau_{t+1}) - \Phi_i(\tau_t) = \sum_{k=0}^{t} s_{i,k} - \sum_{k=0}^{t-1} s_{i,k} = s_{i,t}
\end{equation}
Since $s_{i,t}$ adheres to the PBRS condition, the transformation is policy-invariant. The expected total return in the shaped environment is:
\begin{equation}
    J_{\text{PBRS}}(\pi) = \mathbb{E}_{\pi}\left[\sum_{t,i} (r_{i,t}^{\text{orig}} + s_{i,t}) \right] = J_{\text{env}}(\pi) + \mathbb{E}_{\pi}\left[\sum_{t,i} s_{i,t}\right]
\end{equation}
By construction, the total shaping reward $\sum_{t,i} s_{i,t} = R(s_T)$. Since $J_{\text{env}}(\pi) = \mathbb{E}_{\pi}[R(s_T)]$, the new objective is $J_{\text{PBRS}}(\pi) = 2J_{\text{env}}(\pi)$. As this is a constant scaling of the original objective, the set of optimal policies is preserved.
\end{proof}

\subsection{Analysis of Gradient Dynamics}
\label{subsec:gradient_dir_preservation}
Beyond guaranteeing optimality, we now analyze how TAR² influences the learning dynamics. We prove that while TAR² introduces a beneficial bias to the joint policy gradient, it preserves the gradient direction for each individual agent.

\begin{proposition}[Stochastic Gradient Direction Preservation]
\label{prop:gradient_direction}
For any agent $k$ and any sampled trajectory $\tau$, the stochastic policy gradient estimate under TAR²'s rewards, 
\begin{equation}
    \hat{\mathbf{g}}_{k, \text{TAR}^2}(\tau) = \delta_k(\tau) \hat{\mathbf{g}}_{k, \text{Global}}(\tau),
\end{equation}
is proportional to the gradient estimate under the original team reward, $\hat{\mathbf{g}}_{k, \text{Global}}(\tau)$.
The scaling factor $\delta_k(\tau) \in [0,1]$ is a trajectory-dependent scalar.
\end{proposition}
\begin{proof}
The policy gradient estimate for agent $k$ is the product of its score function and the total return it receives. Under the global episodic reward, $R(s_T)$, the estimate is
\begin{equation}
    \hat{\mathbf{g}}_{k, \text{Global}}(\tau) = G_k(\tau) R(s_T),
\end{equation}
with the score function $G_k(\tau) = \textstyle \sum_{t=0}^{T-1} \nabla_{\theta_k} \log \pi_{\theta_k}(a_{k,t}|\tau_{k,t})$. Under TAR², the estimate uses the agent's individual return, $R_k(\tau) = \textstyle \sum_{t=0}^{T-1} r_{k,t}$, resulting in
\begin{equation}
    \hat{\mathbf{g}}_{k, \text{TAR}^2}(\tau) = G_k(\tau) R_k(\tau).
\end{equation}
By substituting our definition of $r_{k,t}$ from Eq.~\ref{eq:reward_redist}, we find that $R_k(\tau)$ is a scaled version of the global reward
\begin{equation}
    R_k(\tau) = \left( \textstyle \sum_{t=0}^{T-1} w_{k,t}^{\text{agent}} w_{t}^{\text{temp}} \right) R(s_T).
\end{equation}
Letting $\delta_k(\tau) = \textstyle \sum_{t} w_{k,t}^{\text{agent}} w_{t}^{\text{temp}}$, the relationship becomes $\hat{\mathbf{g}}_{k, \text{TAR}^2}(\tau) = \delta_k(\tau) \hat{\mathbf{g}}_{k, \text{Global}}(\tau)$. Since weights $w_{k,t}^{\text{agent}} \in [0,1]$ and $w_{t}^{\text{temp}} \in [0,1]$ with $\sum_t w_{t}^{\text{temp}} = 1$, the scalar $\delta_k(\tau) = \sum_t w_{t}^{\text{temp}} w_{k,t}^{\text{agent}} \le \sum_t w_{t}^{\text{temp}} = 1$. As all weights are non-negative, $\delta_k(\tau) \ge 0$. Thus, $\delta_k(\tau) \in [0,1]$, and the stochastic gradient direction for agent $k$ is preserved.
\qedhere
\end{proof}

\paragraph{Implications for Joint Policy Convergence.} 
Crucially, while TAR² preserves the gradient direction for \textit{each agent individually}, the joint policy gradient, $\mathbf{G}_{\text{TAR}^2} = \sum_k \delta_k(\tau) \hat{\mathbf{g}}_k$, is not parallel to the true joint gradient, $\mathbf{G}_{\text{Global}} = \sum_k \hat{\mathbf{g}}_k$. This deviation introduces a \textit{beneficial bias} -- TAR² trades the unbiased but high-variance true gradient for a lower-variance, biased estimate that credits agents proportionally to their learned contribution. While this informed bias may lead the parameters to a different convergent point $\theta^*_{\text{TAR}^2}$ \citep{devlin2014potential}, our PBRS guarantee ensures the resulting policy, $\pi( ;\theta^*_{\text{TAR}^2})$, remains in the set of optimal policies. This provides a structural robustness that methods reliant on unconstrained regression targets lack.

\subsection{Variance Reduction Properties of TAR²}
\label{subsec:variance_analysis}
While our framework guarantees that the optimal policy is preserved (Sec \ref{subsec:pbrs_guarantee}), this alone does not guarantee efficient learning. The standard PBRS formulation allows for any potential function, which can introduce significant noise and high variance into the policy gradients, leading to slow and unstable convergence. TAR² is explicitly designed to mitigate this issue through two primary mechanisms, which we analyze below.

\subsection{Variance Reduction from Structural Constraints}
\label{subsubsec:strict_return_equivalence}
The core design choice of TAR\textsuperscript{2} is to use only the shaped reward $s_{i,t}$ for learning, rather than the full potential-based reward $r'_{i,t} = r_{i,t}^{\text{orig}} + s_{i,t}$. This choice structurally reduces the variance of the joint policy gradient estimator.

Let the score function for agent $k$ be $G_k(\tau) = \sum_{t=0}^{T-1} \nabla_{\theta_k} \log \pi_{\theta_k}(a_{k,t} | \tau_{k,t})$. The joint policy gradient estimator under a full PBRS formulation would be $\hat{g}_{\text{PBRS}}(\tau) = \sum_k G_k(\tau) (R_k^{\text{orig}}(\tau) + S_k(\tau))$, where $R_k^{\text{orig}}$ and $S_k$ are the returns for agent $k$ from the original and shaped rewards, respectively. This estimator can be decomposed into two parts: a gradient component from the original rewards, $\hat{g}_{\text{orig}}(\tau) = \sum_k G_k(\tau) R_k^{\text{orig}}(\tau)$, and the TAR\textsuperscript{2} estimator from our shaped rewards, $\hat{g}_{\text{TAR}^2}(\tau) = \sum_k G_k(\tau) S_k(\tau)$.

Using the decomposition for the variance of a sum, the variance of the full PBRS estimator is $Var(\hat{g}_{\text{PBRS}}) = Var(\hat{g}_{\text{orig}}) + Var(\hat{g}_{\text{TAR}^2}) + 2Cov(\hat{g}_{\text{orig}}, \hat{g}_{\text{TAR}^2})$. The term $Var(\hat{g}_{\text{orig}})$ is the primary source of instability. In the episodic setting, the per-agent return $R_k^{\text{orig}}(\tau)$ is a sparse, high-variance signal, making $\hat{g}_{\text{orig}}$ a noisy estimator for the joint policy gradient. Since $Var(\hat{g}_{\text{orig}})$ is a non-negative term, it follows that $Var(\hat{g}_{\text{PBRS}}) > Var(\hat{g}_{\text{TAR}^2})$.

By using only the shaped rewards $s_{i,t}$ for learning, TAR\textsuperscript{2} structurally eliminates the noisy estimator component $\hat{g}_{\text{orig}}$. This provides a denser, lower-variance signal for more stable and efficient learning, while the policy-invariance guarantee is retained because the total redistributed reward equals the original team reward ($\sum_{t,i} s_{i,t} = R(s_T)$).

\subsubsection{Variance Reduction from Final-State Conditioning}
\label{subsubsec:variance_reduction_from_final_state}
Second, our use of final-state conditioning provides a more causally-correct and lower-variance learning target. Our reward model predicts contribution scores, $c_{i,t}$, conditioned on the final outcome of the trajectory, $Z$ (e.g., the terminal state). By the Law of Total Variance, the variance of these scores can be decomposed as
\begin{equation}
\label{eq:total_variance_cramped}
    \underbrace{\text{Var}(c_{i,t} | \boldsymbol{\tau}_t)}_{\text{Original Var.}} \! = \! \underbrace{\mathbb{E}[\text{Var}(c_{i,t} | \boldsymbol{\tau}_t, Z)]}_{\text{Remaining Var.}} \! + \! \underbrace{\text{Var}(\mathbb{E}[c_{i,t} | \boldsymbol{\tau}_t, Z])}_{\text{TAR² Target Var.}}
\end{equation}
Our model's learning target is the final term, the expected contribution $\mathbb{E}[c_{i,t} | \boldsymbol{\tau}_t, Z]$. Since variance is non-negative, this proves that the variance of our learning target is less than or equal to the variance of the unconditioned signal
\begin{equation}
    \text{Var}(\mathbb{E}[c_{i,t} | \boldsymbol{\tau}_t, Z]) \le \text{Var}(c_{i,t} | \boldsymbol{\tau}_t).
\end{equation}
By learning this less noisy, post-hoc signal, TAR² benefits from a sharper and more stable learning target.

The introduction of the final outcome, $Z$, means our potential function is implicitly conditioned on information from the end of the trajectory, i.e., $\Phi(\boldsymbol{\tau}_t, Z)$. This does not invalidate the PBRS guarantees. The policy invariance proof relies on the telescoping sum of potential differences, which holds even when the potential is conditioned on a variable, $Z$, that is constant for any given trajectory \citep{devlin2014potential, arjona2019rudder}. Since this is a valid choice for the potential function's state representation, all policy preservation guarantees remain intact.

\subsection{Model Architecture and Training}
\label{subsec:architecture}
We now detail the architecture that operationalizes our TAR² method. As depicted in Figure~\ref{fig:tar2_architecture}, our model is a sequence-to-sequence network that processes the joint action-observation history, $\boldsymbol{\tau}$, to produce the unnormalized contribution scores, $c_{i,t}$.

\subsubsection{Architecture Details}
\label{subsubsec:arch_details}
Our reward model uses a dual-transformer design to capture both temporal and inter-agent dependencies. As shown in Figure~\ref{fig:tar2_architecture}, the architecture integrates three key components:
\paragraph{Final-State Conditioning.}
To ground credit assignment in the trajectory's outcome, the entire model is conditioned on an embedding of the final state, $Z$. As justified in our variance analysis (Sec \ref{subsubsec:variance_reduction_from_final_state}), this provides a provably lower-variance learning target.

\paragraph{Inverse Dynamics for Causal Representations.}
To ensure the embeddings are causally relevant, an auxiliary inverse dynamics model regularizes the shared representations and improve downstream task performance \citep{pathak2017curiosity, brandfonbrener2023inversedynamicspretraininglearns}. A separate MLP prediction head takes the concatenated latent embeddings from two consecutive timesteps, $\text{embed}(s_{i,t})$ and $\text{embed}(s_{i,t+1})$, and outputs a probability distribution, $\pi_{\text{ID}}$, over the action space. It is trained via an auxiliary loss to predict the agent's action, $a_{i,t}$. This forces the shared embeddings to encode controllable aspects of the state, preventing the credit assignment from relying on spurious correlations.

\paragraph{Deterministic Reward Normalization.}
Finally, the raw scores $c_{i,t}$ from the transformer body are passed to the deterministic normalization function (Sec \ref{subsec:rew_red_form}). This non-learned step is the structural property that guarantees the strict return equivalence required for effectively learning and preserving the environment's optimal policy.

\subsubsection{Training Objective}
\label{subsubsec:training_objective}
The reward model, parameterized by $\theta$, is trained by minimizing a composite loss function. The objective provides a strong learning signal for identifying the \textit{importance} of each agent's actions by regressing the sum of scores against the true episodic reward, while an inverse dynamics term regularizes the representations. The resulting objective
\vspace{-0.25em}
\begin{equation}
    \mathcal{L}(\theta) = \mathbb{E}_{\tau \sim \mathcal{B}} \bigg[ \underbrace{
        \left( R(s_T) - \textstyle\sum_{t,i} c_{i,t} \right)^2
    }_{\text{Reward Regression Loss}} \underbrace{
        - \lambda \textstyle\sum_{t,i} \log \pi_{\text{ID}}(a_{i,t})
    }_{\text{Inverse Dynamics Regularizer}} \bigg]
    \label{eq:training_objective}
\end{equation}
combines this regression loss with the inverse dynamics regularizer, modulated by a hyperparameter $\lambda$. This objective trains the model to assign a higher total score to better trajectories. The separate normalization step (Sec \ref{subsec:rew_red_form}) then takes these meaningfully-scaled scores and distributes their value as credit in a way that structurally guarantees policy invariance.

\begin{figure*}[!ht]
    \centering
    \includegraphics[width=0.75\linewidth]{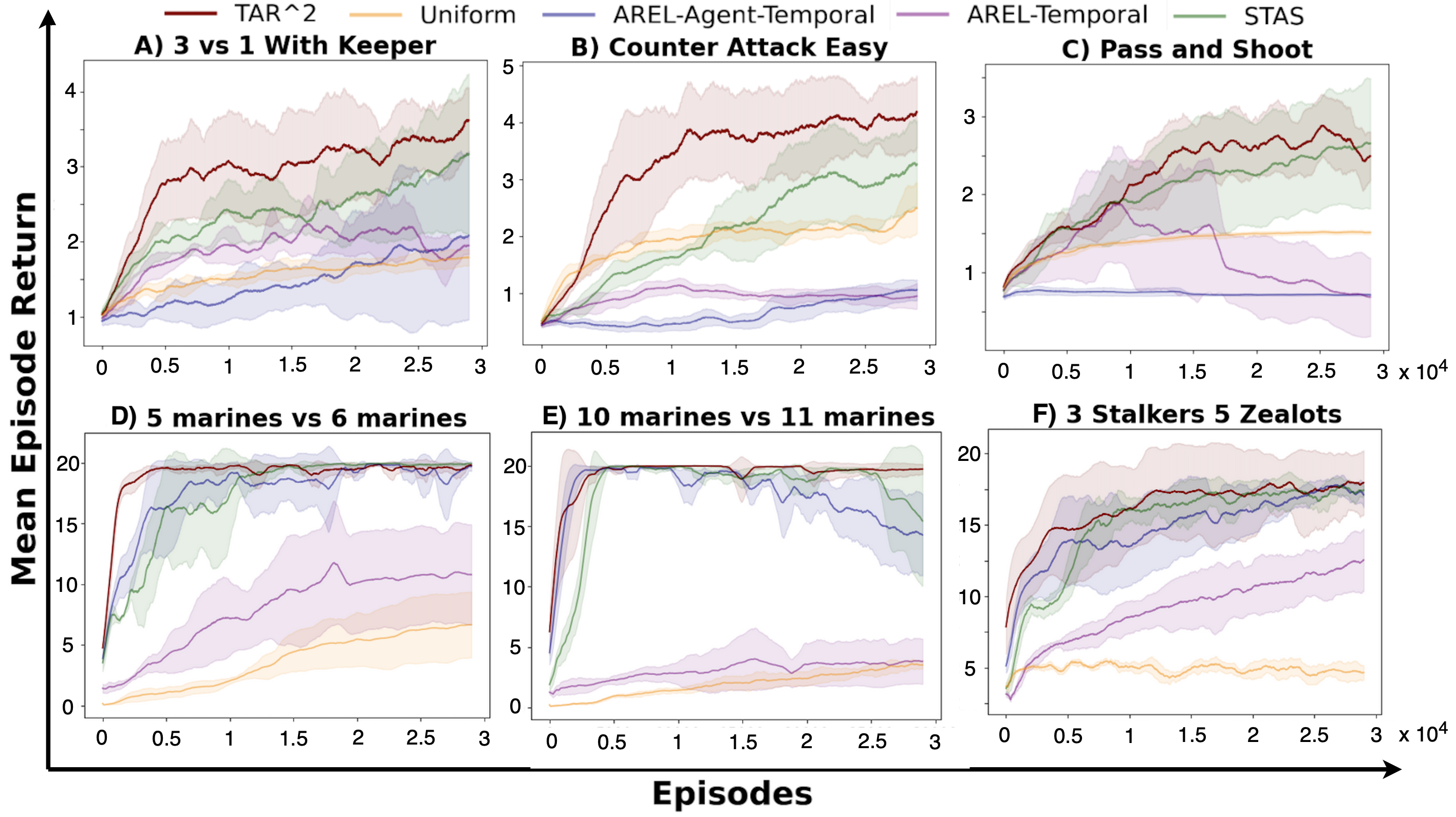}
    \caption{TAR²'s Average Return comparison against baselines on Google Research Football (A-C) and SMACLite (D-F). On SMACLite, it demonstrates improved sample efficiency compared to STAS and converges to a higher average return than the unstable AREL variants. This trend is more pronounced on GRF, where TAR² consistently achieves the highest average return, particularly in `Counter Attack Easy` and `Pass and Shoot` scenarios.}
    \label{fig:baselines}
\end{figure*}

\begin{figure*}[!ht]
    \centering
    \includegraphics[width=0.75\linewidth]{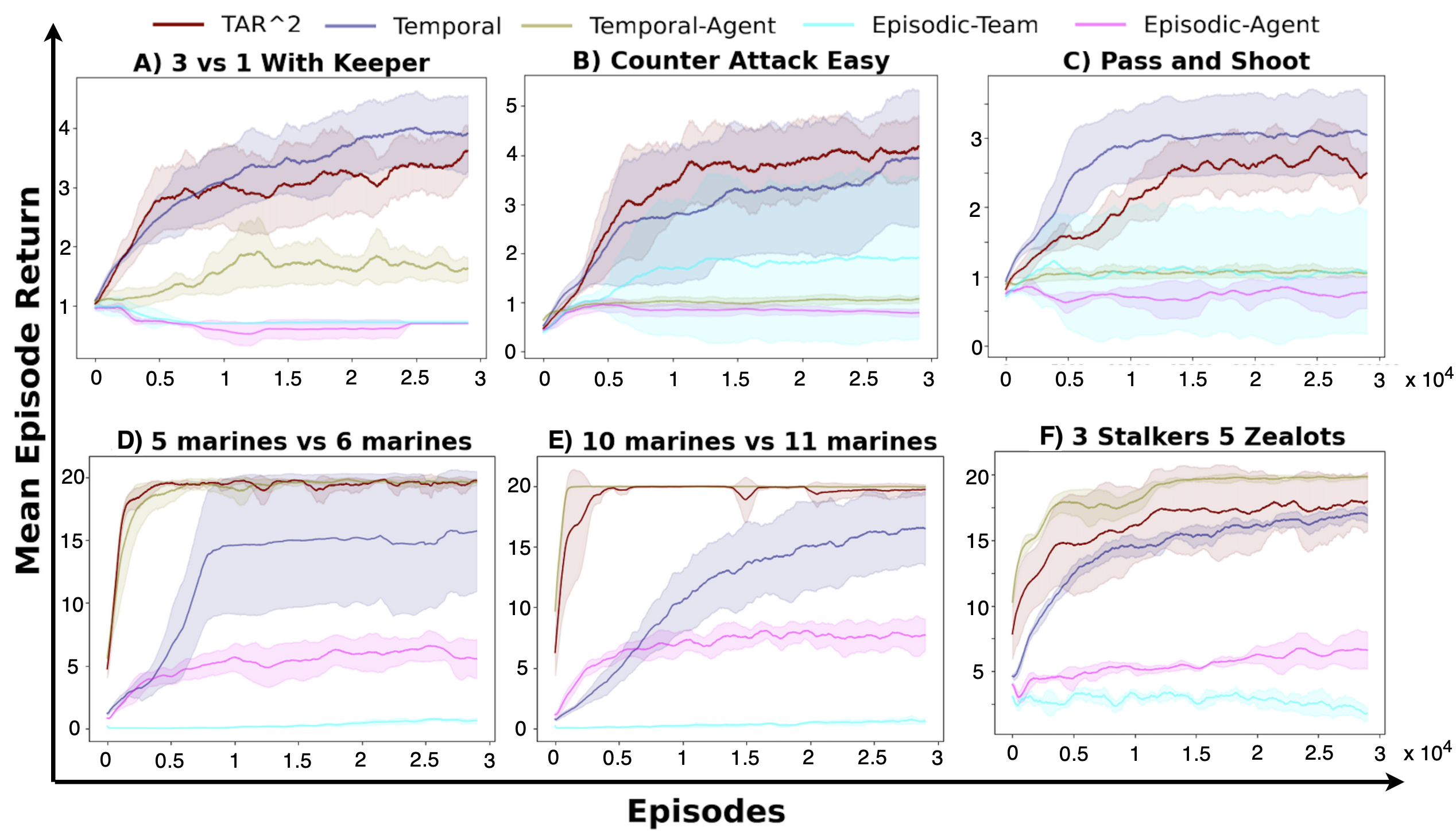}
    \caption{Performance of TAR² relative to oracle rewards on SMACLite (D-F) and Google Research Football (A-C). TAR² enables learning a policy that is competitive with hand-crafted reward functions. TAR²'s performance rivals `Temporal-Agent` in SMACLite and `Temporal` in GRF. It outperforms all other heuristics, demonstrating that a learned credit assignment can be more effective than a manually engineered one.}
    \label{fig:bounds}
\end{figure*}

\begin{figure*}[!ht]
    \centering
    \includegraphics[width=0.6\linewidth]{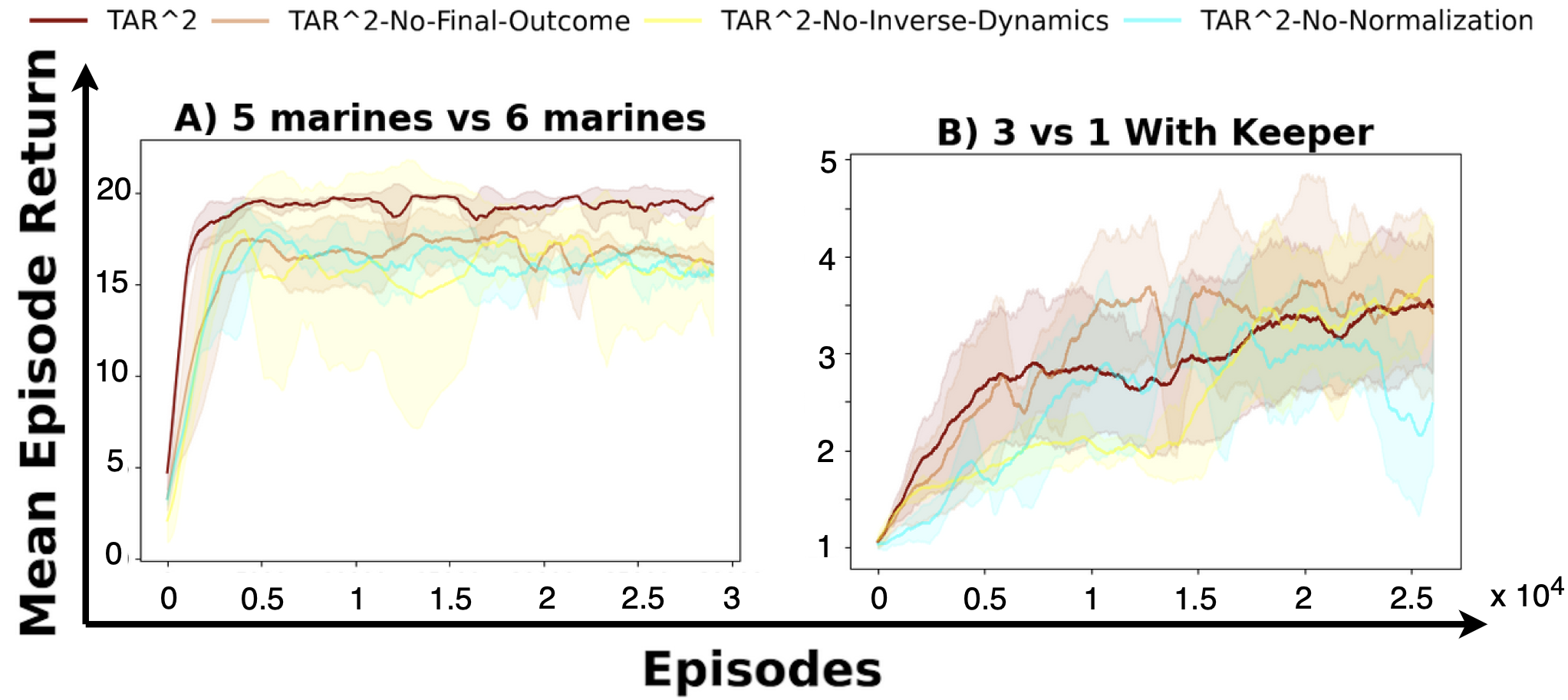}
    \caption{Ablation study of TAR²'s core components. Removing any component degrades performance. `No-Final-Outcome` increases variance, `No-Inverse-Dynamics` hinders performance, and `No-Normalization` is the most detrimental as it violates the policy preservation guarantee.}
    \label{fig:ablations}
\end{figure*}

\section{Related Work}
\label{sec:related_works}
Our work addresses the joint agent-temporal credit assignment problem. We position TAR² within the existing literature, highlighting how our design overcomes the key limitations of prior work, particularly their reliance on theoretically brittle learning schemes.

\subsection{Temporal Credit Assignment}
\label{subsec:temporal_credit_assignment}
Temporal credit assignment aims to transform a single, sparse episodic reward into a sequence of per-timestep signals. Prominent single-agent approaches include analyzing state-value differences (RUDDER), using sequence models or retrospectively re-evaluating actions \citep{arjona2019rudder, liu2019sequence, harutyunyan2019hindsight}, while others use \textit{intrinsic motivation} to generate dense rewards for exploration \citep{schaefer2022derl}. These methods, however, have two key limitations in our context. First, they are designed for single-agent problems and do not address the multi-agent nature of credit assignment. Second, approaches like intrinsic motivation intentionally alter the optimization objective to encourage exploration, whereas our goal is to preserve the original optimal policy. Even methods adapted for MARL, such as AREL \citep{xiao2022agent}, address temporal credit for the team as a whole but fail to disentangle individual agent contributions.

\subsection{Agent Credit Assignment}
\label{subsec:agent_credit_assignment}
Agent credit assignment allocates a shared team reward to individual agents, dominated by methods like value function factorization (VDN, QMIX) and counterfactuals (COMA) \citep{sunehag2017value, rashid2020monotonic, foerster2018counterfactual}. Other approaches use Shapley values or attention-based critics \citep{wang2020shapley, freed2021learning, kapoor2024assigning}. A unifying limitation of these methods is their fundamental reliance on dense, per-timestep team rewards, an assumption that fails in the challenging episodic settings we address \citep{papoudakis2021benchmarking}. In contrast, TAR² is designed specifically for a single episodic signal.

\subsection{Joint Agent-Temporal Credit Assignment}
\label{subsec:agent_temporal_credit_assignment}
Addressing both credit assignment dimensions is a key frontier in MARL. The most notable prior work, STAS \citep{Chen2023STASSR}, tackles this joint problem but its approach is theoretically brittle. It trains a model to directly predict the final shaped rewards, meaning its policy preservation guarantee is conditional on the network's regression accuracy. This reliance on perfect function approximation is the key theoretical fragility that TAR² is designed to overcome, as we detail in Sec \ref{sec:approach}.

\section{Experiments}
\label{sec:experiments}
We conduct experiments to answer three core questions: (\textbf{1}) How does TAR² compare to other reward redistribution baselines across diverse and challenging environments? (\textbf{2}) How does its performance compare against oracles with access to privileged reward information? (\textbf{3}) Which of our architectural components are most critical to its success?

\subsection{Experimental Setup}
\label{subsec:experiment_setup}

\paragraph{Environments.}
We evaluate on two challenging benchmarks, SMACLite~\citep{michalski2023smaclite} and Google Research Football (GRF)~\citep{kurach2020google}, modifying both to be strictly episodic with a single terminal team reward. Our chosen maps test distinct coordination challenges: SMACLite's \textit{5m vs 6m} and \textit{10m vs 11m} test scalability, while \textit{3s5z} tests coordination between heterogeneous agents. For GRF, we use \textit{3 vs 1 with keeper}, \textit{counterattack easy}, and \textit{pass and shoot} to assess performance across diverse strategies. Further details are in the appendix.

\paragraph{Baselines.}
For a fair comparison, all methods are built upon the same state-of-the-art MAPPO implementation \citep{yu2022surprising}. We compare against several credit assignment frameworks that operate on the same episodic team reward. These include a naive Uniform credit baseline; AREL-Temporal \citep{xiao2022agent} for temporal-only assignment; and two strong joint agent-temporal methods, STAS \citep{Chen2023STASSR}, the current state-of-the-art, and our adapted AREL-Agent-Temporal. We omit value decomposition methods like QMIX as they rely on Temporal-Difference (TD) updates, which are ill-suited for settings with only a single, delayed episodic reward \citep{gangwani2020learning}.

\paragraph{Oracle Baselines.}
To contextualize TAR²'s performance, we establish oracle baselines using privileged information from the original, dense-reward environments. The original environments provide a dense, per-timestep team reward, which we term Temporal. We then create a factorized version of this reward that provides a per-timestep, per-agent signal, termed Temporal-Agent. For our main experiments, the standard input for TAR² and all baselines is the Episodic-Team reward, which is the sum of Temporal rewards over an episode. Similarly, Episodic-Agent reward is the sum of Temporal-Agent rewards. This allows us to compare against oracles trained with more information, thereby establishing heuristic performance bounds.

\subsection{Results and Discussion}
\label{subsec:results_and_discussion}

All learning curves show the mean episode return over 5 random seeds, with shaded areas representing 95\% confidence intervals. We use average return as our primary metric to clearly analyze performance gains throughout the entire training process.

\paragraph{Q1: Performance Against Baselines.}
As shown in Figure~\ref{fig:baselines}, TAR² consistently outperforms all baselines in both final average return and sample efficiency. The results empirically validate our core thesis: the structural robustness of TAR² provides a more reliable learning signal than the theoretically brittle designs of prior work. For instance, while STAS is competitive in some SMACLite scenarios, it exhibits higher variance and is unstable in GRF. This aligns with our analysis that its reliance on direct reward regression is fragile. The poor performance of the AREL variants further supports this conclusion. The performance of the Uniform baseline is particularly telling, it fails completely in SMACLite, yet is surprisingly effective in GRF, outperforming more complex methods. This highlights that without a robust theoretical grounding, even advanced models can struggle to beat simple heuristics. TAR²'s stable learning across all challenging environments demonstrates the practical benefit of decoupling credit modeling from constraint satisfaction.

\paragraph{Q2: Contextualizing Performance with Oracle Bounds.}
Figure~\ref{fig:bounds} contextualizes TAR²'s performance by comparing it against oracles with access to privileged reward signals. The results show that TAR² learns a highly effective credit assignment strategy using only the sparse Episodic-Team reward. In SMACLite, particularly on the homogeneous maps (\textit{5m vs 6m}, \textit{10m vs 11m}), TAR²'s performance is indistinguishable from an oracle trained on perfect, per-agent dense rewards (Temporal-Agent signal).Even on the heterogeneous map (\textit{3s5z}), TAR² learns a near-optimal credit distribution, significantly outperforming weaker oracle signals. In GRF, TAR² is consistently the best-performing non-oracle method and is even superior to all the oracle heuristics in the \textit{Counter Attack Easy} scenario. These results indicate that TAR² learns a credit distribution that is highly competitive with, and at times better than, what can be achieved with hand-crafted, privileged reward functions.

\paragraph{Q3: Ablation Studies of Architectural Components.}
Our ablation studies (Figure~\ref{fig:ablations}) confirm that each component of our design is critical to its success. Removing Final Outcome Conditioning increases learning variance, which is consistent with our analysis in Sec \ref{subsubsec:variance_reduction_from_final_state}. Ablating the Inverse Dynamics regularizer degrades performance, confirming that grounding the representations in causal actions is crucial \citep{pathak2017curiosity, brandfonbrener2023inversedynamicspretraininglearns}. Learning less meaningful state representations leads to faulty credit assignment and instability, a known failure mode~\citep{kapoor2024assigning} due to imprecise credit assignment. The most significant performance collapse occurs when removing our Deterministic Normalization function. This ablation forces the model to suffer from the same fundamental flaw as AREL and STAS: optimizing an objective with no structural policy preservation guarantee. The resulting instability is a direct consequence of this brittle design. Collectively, these studies provide strong evidence that our components work synergistically to reduce variance, promote causal representations, and structurally guarantee policy preservation.

\section{Limitations and Future Work}
\label{sec:lim_future}
While TAR² establishes a robust framework for episodic credit assignment, several exciting avenues for future work remain. First, the transformer-based architecture may face scalability challenges and could be enhanced by exploring methods like sparse attention or explicit group decomposition to better model agent interactions in massive-scale systems. Second, our framework is currently designed for a single terminal reward; a key next step is to extend TAR² to handle scenarios with multiple sparse rewards within an episode, which would require adapting our formulation. Furthermore, the residual performance variance observed in our results may stem from the implicit exploration encouraged by PBRS \citep{devlin2014potential}. A formal exploration-exploitation analysis could lead to adaptive shaping strategies. A more advanced approach would be to frame this as a bi-level optimization problem, where the reward model is meta-learned to produce shaping signals that directly maximize the downstream performance improvement of the agent policies. Finally, TAR²'s ability to learn from a single, outcome-based signal makes it a prime candidate for training teams of multi-agent Large Language Models (LLMs), a domain where feedback is often sparse.

\section{Conclusion}
\label{sec:conclusion}
We introduced TAR², a method for joint agent-temporal credit assignment designed for structural robustness. By decoupling credit modeling from constraint satisfaction, our approach overcomes the theoretical fragility of prior methods that rely on unreliable reward regression. The method's deterministic normalization step guarantees strict return equivalence by construction, which we prove is equivalent to a valid Potential-Based Reward Shaping (PBRS), ensuring the optimal policy is preserved. Experiments on challenging SMACLite and GRF scenarios show that our approach learns faster and achieves better final performance than state-of-the-art baselines. These results validate our design principle and establish a robust and theoretically-grounded method for credit assignment in complex episodic MARL tasks.

\bibliography{iclr2026_conference}
\bibliographystyle{iclr2026_conference}

\appendix

\section{Detailed Task Descriptions}
\label{sec:detailed_task_description}
\paragraph{SMACLite \citep{michalski2023smaclite}}
A computationally efficient variant of StarCraft II \citep{samvelyan2019starcraft}. We experiment on three battle scenarios with varying complexity:
\begin{description}
    \item[5m\_vs\_6m \& 10m\_vs\_11m] Homogeneous scenarios testing scalability.
    \item[3s5z] A heterogeneous scenario with 3 Stalkers and 5 Zealots, testing coordination between different unit types.
\end{description}
Each agent’s local observation includes the relative positions, unit types, health, and shield strength of allies and enemies within its field of view, as well as its own status. Agents can move, stop, or attack any visible enemy. The environment provides action masks for valid actions. Each combat scenario lasts for 100 timesteps. The environment’s reward function combines partial rewards for damaging or eliminating enemies, with a maximum possible team return normalized to 20. The repository is available at: \url{https://github.com/uoe-agents/smaclite} (MIT License).

\paragraph{Google Research Football (GRF) \citep{kurach2020google}}
A high-fidelity multi-agent football (soccer) simulation. We evaluate on three scenarios:
\begin{description}
    \item[academy\_3\_vs\_1\_with\_keeper] A basic offensive scenario.
    \item[academy\_counterattack\_easy] Tests rapid transitions from defense to offense.
    \item[academy\_pass\_and\_shoot\_with\_keeper] Requires precise passing and coordination to score.
\end{description}
The observation space (`simple115v2`) includes player positions, ball coordinates, velocity vectors, and stamina. The action space includes passing, shooting, dribbling, and tackling. Episodes end after a goal or 200 timesteps. Reward signals are sparse, tied to events like scoring goals or advancing the ball. The repository is available at: \url{https://github.com/google-research/football} (Apache License 2.0).

\section{Implementation Details and Hyperparameters}
\label{sec:hyperparameters}

The code was run on Lambda Labs deep learning workstation with 2-4 Nvidia RTX 2080 Ti graphics cards.  Each training run was run on one single GPU, and required approximately 10 hours. 

 Hyperparameters used for TAR$^2$, STAS, AREL-Temporal, AREL-Agent-Temporal, Uniform and various environment reward configurations that are common to all tasks are shown in Tables~\ref{tab:mappo_common_hyperparams}. The task-specific hyperparameters considered in our grid search for TAR$^2$, STAS, AREL-variants in Tables \ref{table:tar^2_hyperparameter_sweep}, \ref{table:stas_hyperparameter_sweep} and \ref{table:arel_hyperparameter_sweep} respectively. Bold values indicate the optimal hyperparameters.

\begin{table}[H]
\centering
\caption{Common Hyperparameters for MAPPO algorithms.}
\begin{tabular}{|c|c|}
\hline
\textbf{\begin{tabular}[c]{@{}c@{}}common\\ hyperparameters\end{tabular}} & \textbf{value} \\ \hline
ppo\_epochs         & 15 \\   \hline
ppo\_batch\_size      & 30 \\   \hline
gamma               & 0.99 \\ \hline
max\_episodes       & 30000 \\ \hline
max\_time\_steps    & 100 \\  \hline
rnn\_num\_layers\_v    & 1 \\  \hline
rnn\_hidden\_v      & 64  \\  \hline
v\_value\_lr          & 5e-4  \\  \hline
v\_weight\_decay       & 0.0 \\ \hline
v\_hidden\_shape    & 64  \\  \hline
grad\_clip\_critic\_v  & 0.5  \\  \hline
value\_clip         & 0.2  \\  \hline
data\_chunk\_length & 10  \\  \hline
rnn\_num\_layers\_actor    & 1 \\  \hline
rnn\_hidden\_actor      & 64  \\  \hline
policy\_lr          & 5e-4  \\  \hline
policy\_weight\_decay       & 0.0 \\ \hline
policy\_hidden\_shape    & 64  \\  \hline
grad\_clip\_actor  & 0.5  \\  \hline
policy\_clip         & 0.2  \\  \hline
entropy\_pen       & 1e-2 \\  \hline
gae\_lambda         & 0.95 \\ \hline
\end{tabular}

\label{tab:mappo_common_hyperparams}
\end{table}

\begin{table}[!h]
\centering
\caption{TAR$^2$ hyperparameters.}
\label{table:tar^2_hyperparameter_sweep}
\footnotesize  
\begin{tabularx}{\textwidth}{|*{14}{>{\centering\arraybackslash}X|}}
\hline
\textbf{Env. Name} & \textbf{num heads} & \textbf{depth} & \textbf{dropout} & \textbf{comp. dim} & \textbf{batch size} & \textbf{lr} & \textbf{weight decay} & \textbf{inv. dyn. loss coef.} & \textbf{grad clip val.} & \textbf{model upd. freq.} & \textbf{model upd. epochs} & \textbf{policy lr} & \textbf{entropy coef} \\
\hline
Google Football & [3, \textbf{4}] & [\textbf{3}, 4] & [\textbf{0.0}, 0.1, 0.2] & [16, \textbf{64}, 128] & [32, 64, \textbf{128}] & [1e-4, \textbf{5e-4}, 1e-3] & [\textbf{0.0}, 1e-5, 1e-4] & [1e-3, 1e-2, \textbf{5e-2}] & [0.5, 5.0, \textbf{10.0}] & [50, 100, \textbf{200}] & [100, \textbf{200}, 400] & [5e-4, \textbf{1e-3}] & [5e-3, 8e-3, \textbf{1e-2}] \\
\hline
SMACLite & [3, \textbf{4}] & [\textbf{3}, 4] & [\textbf{0.0}, 0.1, 0.2] & [16, \textbf{64}, 128] & [32, 64, \textbf{128}] & [1e-4, \textbf{5e-4}, 1e-3] & [\textbf{0.0}, 1e-5, 1e-4] & [1e-3, 1e-2, \textbf{5e-2}] & [0.5, 5.0, \textbf{10.0}] & [50, 100, \textbf{200}] & [100, \textbf{200}, 400] & [5e-4, \textbf{1e-3}] & [\textbf{5e-3}, 8e-3, 1e-2] \\
\hline
\end{tabularx}
\end{table}

\begin{table}[!h]
\centering
\caption{STAS hyperparameters.}
\label{table:stas_hyperparameter_sweep}
\footnotesize  
\begin{tabularx}{\textwidth}{|*{14}{>{\centering\arraybackslash}X|}}
\hline
\textbf{Env. Name} & \textbf{num heads} & \textbf{depth} & \textbf{dropout} & \textbf{comp. dim} & \textbf{batch size} & \textbf{lr} & \textbf{weight decay} & \textbf{grad clip val.} & \textbf{model upd. freq.} & \textbf{model upd. epochs} \\
\hline
Google Football & [3, \textbf{4}] & [\textbf{3}, 4] & [\textbf{0.0}, 0.1, 0.2] & [16, \textbf{64}, 128] & [32, 64, \textbf{128}] & [1e-4, \textbf{5e-4}, 1e-3] & [\textbf{0.0}, 1e-5, 1e-4] & [0.5, 5.0, \textbf{10.0}] & [50, 100, \textbf{200}] & [100, \textbf{200}, 400] \\
\hline
SMACLite & [3, \textbf{4}] & [\textbf{3}, 4] & [\textbf{0.0}, 0.1, 0.2] & [16, \textbf{64}, 128] & [32, 64, \textbf{128}] & [1e-4, \textbf{5e-4}, 1e-3] & [\textbf{0.0}, 1e-5, 1e-4] & [0.5, 5.0, \textbf{10.0}] & [50, 100, \textbf{200}] & [100, \textbf{200}, 400] \\
\hline
\end{tabularx}
\end{table}

\begin{table}[!h]
\centering
\caption{AREL hyperparameters.}
\label{table:arel_hyperparameter_sweep}
\footnotesize  
\begin{tabularx}{\textwidth}{|*{14}{>{\centering\arraybackslash}X|}}
\hline
\textbf{Env. Name} & \textbf{num heads} & \textbf{depth} & \textbf{dropout} & \textbf{comp. dim} & \textbf{batch size} & \textbf{lr} & \textbf{weight decay} & \textbf{grad clip val.} & \textbf{model upd. freq.} & \textbf{model upd. epochs} \\
\hline
Google Football & [3, \textbf{4}] & [\textbf{3}, 4] & [\textbf{0.0}, 0.1, 0.2] & [16, \textbf{64}, 128] & [32, 64, \textbf{128}] & [1e-4, \textbf{5e-4}, 1e-3] & [\textbf{0.0}, 1e-5, 1e-4] & [0.5, 5.0, \textbf{10.0}] & [50, 100, \textbf{200}] & [100, \textbf{200}, 400] \\
\hline
SMACLite & [3, \textbf{4}] & [\textbf{3}, 4] & [\textbf{0.0}, 0.1, 0.2] & [16, \textbf{64}, 128] & [32, 64, \textbf{128}] & [1e-4, \textbf{5e-4}, 1e-3] & [\textbf{0.0}, 1e-5, 1e-4] & [0.5, 5.0, \textbf{10.0}] & [50, 100, \textbf{200}] & [100, \textbf{200}, 400] \\
\hline
\end{tabularx}
\end{table}

\section{Pseudocode}
\label{sec:pseudocode}
Our training process is detailed in the following algorithms. Algorithm~\ref{alg:main_loop} describes the main on-policy training loop which collects data and updates the MAPPO actor and critic policies. Algorithm~\ref{alg:tar2_training} describes the periodic, off-policy training of the TAR² reward model.

\begin{algorithm}[h!]
\caption{Main Training Loop: MAPPO with TAR² Rewards}
\label{alg:main_loop}
\begin{algorithmic}[1]
\STATE \textbf{Initialize:} Policy nets $\pi_{\omega_i}$, critic nets $V_{\mu_i}$ for each agent $i=1..N$.
\STATE \textbf{Initialize:} TAR² reward model $R_{\theta}$, experience buffer $\mathcal{B} \leftarrow \emptyset$.
\STATE \textbf{Initialize:} PopArt parameters for value normalization.

\WHILE{not converged}
    \STATE Initialize temporary data buffer $D \leftarrow \emptyset$.
    \FOR{$k = 1$ to \textit{num\_rollout\_threads}}
        \STATE Initialize actor RNN hidden states $\mathbf{h}_{0, \pi}$ and critic RNN hidden states $\mathbf{h}_{0, V}$.
        \STATE Initialize empty trajectory storage $\tau_{\text{storage}} \leftarrow []$.
        \FOR{$t = 0$ to $T-1$}
            \STATE Get joint observation $\mathbf{o}_t$.
            \FOR{each agent $i=1..N$}
                \STATE Sample action $a_{i,t}$ and get next hidden state $h_{t+1, \pi}^{(i)}$ from policy:
                \STATE $a_{i,t}, h_{t+1, \pi}^{(i)} \leftarrow \pi_{\omega_i}(o_{i,t}, h_{t, \pi}^{(i)})$.
                \STATE Get state value $v_{i,t}$ and next hidden state $h_{t+1, V}^{(i)}$ from critic:
                \STATE $v_{i,t}, h_{t+1, V}^{(i)} \leftarrow V_{\mu_i}(\mathbf{s}_t, h_{t, V}^{(i)})$, where $\mathbf{s}_t$ is the centralized state representation.
            \ENDFOR
            \STATE Execute joint action $\mathbf{a}_t$, observe next joint observation $\mathbf{o}_{t+1}$.
            \STATE Store transition $(\mathbf{o}_t, \mathbf{a}_t, \mathbf{h}_{t+1, \pi}, \mathbf{h}_{t+1, V}, \{v_{i,t}\}_{i=1}^N)$ in $\tau_{\text{storage}}$.
        \ENDFOR
        \STATE Receive the true episodic team reward $R(s_T)$.
        \STATE Store the completed trajectory $(\tau_{\text{storage}}, R(s_T))$ in the long-term experience buffer $\mathcal{B}$.

        \STATE // --- On-Policy Return and Advantage Calculation ---
        \STATE Compute shaped rewards $\{r_{i,t}\}$ for the trajectory using the TAR² model $R_{\theta}$ and $R(s_T)$.
        \STATE Compute shaped returns $\{G_{i,t}\}$ for each agent using the shaped rewards $\{r_{i,t}\}$.
        \STATE Update PopArt statistics with the shaped returns $\{G_{i,t}\}$.
        \STATE Normalize returns with PopArt.
        \STATE De-normalize value estimates $\{v_{i,t}\}$ using PopArt.
        \STATE Compute advantage estimates $\{\hat{A}_{i,t}\}$ for each agent using GAE.
        \STATE Split trajectory $\tau$ into chunks of length $L$
        \FOR{$l = 0, 1, \ldots, T//L$}
            \STATE $D = D \cup (\tau[l : l + T], \hat{A}[l : l + L], G[l : l + L], \bar{G}[l : l + L])$
        \ENDFOR
        \STATE Add processed chunks to the data buffer $D$.
    \ENDFOR

    \STATE // --- On-Policy Policy and Critic Updates ---
    \FOR{$e = 1$ to \textit{ppo\_epochs}}
        \FOR{mini-batch $b$ sampled from $D$}
            \STATE Update policy parameters $\boldsymbol{\omega}$ using the MAPPO policy loss on advantage estimates from batch $b$.
            \STATE Update critic parameters $\boldsymbol{\mu}$ by regressing on the PopArt-normalized shaped returns from batch $b$.
        \ENDFOR
    \ENDFOR

    \STATE // --- Off-Policy Reward Model Update ---
    \IF{training condition met}
        \STATE Update TAR² reward model parameters $\theta$ using Algorithm~\ref{alg:tar2_training}.
    \ENDIF
\ENDWHILE
\end{algorithmic}
\end{algorithm}

\begin{algorithm}[h!]
\caption{Training the TAR² Reward Model (Off-Policy)}
\label{alg:tar2_training}
\begin{algorithmic}[1]
\STATE \textbf{Input:} Reward model parameters $\theta$, experience buffer $\mathcal{B}$, batch size $B_r$, learning rate $\alpha_R$.
\FOR{$u = 1$ to \textit{num\_reward\_updates}}
    \STATE Sample a batch of trajectories $\{(\tau_j, R(s_T)_j)\}_{j=1}^{B_r}$ from the long-term buffer $\mathcal{B}$.
    \STATE Initialize loss $\mathcal{L}(\theta) \leftarrow 0$.
    \FOR{each trajectory $(\tau, R(s_T))$ in the batch}
        \STATE Compute unnormalized scores $\{c_{i,t}\}$ and predicted actions $\{\hat{a}_{i,t}\}$ from the reward model $R_{\theta}(\tau)$.
        \STATE $\mathcal{L}_{\text{reg}} = \left( \log R(s_T) - \sum_{t,i} c_{i,t} \right)^2$.
        \STATE $\mathcal{L}_{\text{ID}} = - \sum_{t,i} \log \pi_{\text{ID}}(a_{i,t} | \hat{a}_{i,t})$.
        \STATE $\mathcal{L}(\theta) \mathrel{+}= \mathcal{L}_{\text{reg}} + \lambda \mathcal{L}_{\text{ID}}$.
    \ENDFOR
    \STATE Update $\theta$ using Adam optimizer: $\theta \leftarrow \theta - \alpha_R \nabla_{\theta} \mathcal{L}(\theta)$.
\ENDFOR
\end{algorithmic}
\end{algorithm}

\subsubsection{Inverse Dynamics for Causal Representation Learning}
\label{subsubsec:inverse_dynamics_implementation}

To ensure that the learned representations are grounded in agent behavior, we integrate an auxiliary inverse dynamics task into our framework. The goal of this task is to regularize the shared embeddings by forcing them to encode information about the actions that cause transitions. Our implementation is designed to leverage the rich, contextualized embeddings produced by our main architecture.

Our credit assignment model uses a dual temporal-agent attention network to produce global state-action embeddings for each timestep. These embeddings capture not only the state of the environment but also the interactions between agents over time. To form the input for our inverse dynamics prediction, for each agent $i$ at each timestep $t$, we create a concatenated vector, $\mathbf{z}_{i,t}$, from three distinct sources:
\begin{enumerate}
    \item The global state embedding at the current timestep, $\text{embed}(s_t)$.
    \item The global state embedding at the next timestep, $\text{embed}(s_{t+1})$.
    \item The agent-specific state-action embedding from the \textit{previous} timestep, $\text{embed}(s_{i,t-1}, a_{i,t-1})$, which is the output of the dual attention block for that agent.
\end{enumerate}
This concatenated vector $\mathbf{z}_{i,t}$ is then passed through a multi-layer perceptron (MLP), which we refer to as the inverse dynamics head.

Crucially, the network is not trained to predict the action that was actually executed, but rather the action that the agent's policy, $\pi_i$, predicted at that timestep. The objective is to minimize the cross-entropy between the output of the inverse dynamics head and the action probabilities from the policy network. This encourages consistency between the representations used for credit assignment and the representations used for action selection, ensuring the credit is assigned based on features that are directly relevant to the agent's decision-making process.

\section{Interpretability and Insights}
\label{sec:interpretability}

Although our primary focus is on performance and theoretical properties, TAR\(^2\)’s per-timestep, per-agent reward predictions lend themselves to partial interpretability:
\begin{itemize}
    \item Agents’ importance at specific timesteps can be visualized by examining \(w^{\text{temporal}}_t w^{\text{agent}}_{i,t}\).
    \item Comparing predicted reward distributions across different episodes can hint at consistent agent roles or strategic pivot points in the trajectory.
\end{itemize}
However, direct interpretability is challenging in high-dimensional multi-agent environments like SMACLite and Google Football, where the intricate interactions and vast state-action spaces complicate simple visualizations. Additionally, developing a systematic interpretability study would require significant additional methodologies and resources, extending beyond the scope of the current work. While we recognize the importance of interpretability and plan to explore it in future research, our current focus remains on establishing robust performance improvements and theoretical guarantees for reward redistribution.

\end{document}

%% file: math_commands.tex
\usepackage{amsmath,amsfonts,bm}

\def\eqref#1{equation~\ref{#1}}

\def\1{\bm{1}}

\DeclareMathAlphabet{\mathsfit}{\encodingdefault}{\sfdefault}{m}{sl}
\SetMathAlphabet{\mathsfit}{bold}{\encodingdefault}{\sfdefault}{bx}{n}

